\documentclass[aps,amsmath,amssymb,twocolumn,nofootinbib,pra,superscriptaddress]{revtex4}

\usepackage{graphicx}
\usepackage{dcolumn}
\usepackage{bm}
\usepackage{epsfig,hyperref,amsthm}
\usepackage{mathtools}
\usepackage{color}
\usepackage{colordvi}
\usepackage[dvipsnames]{xcolor}
\usepackage{relsize}
\usepackage{accents}
\usepackage{wasysym}  
\usepackage{xypic}

\hypersetup{
	colorlinks=true,
	linkcolor=blue,
	citecolor=blue,
	urlcolor=blue
}

\newcommand{\ket}[1]{\ensuremath{|#1\rangle}}
\newcommand{\bra}[1]{\ensuremath{\langle #1|}}

\newcommand{\proj}[1]{\ket{#1}\bra{#1}}
\newcommand{\be}{\begin{equation}}
\newcommand{\ee}{\end{equation}}
\newcommand{\ba}{\begin{eqnarray}}
\newcommand{\ea}{\end{eqnarray}}

\newcommand{\pe}{\Lambda}

\newcommand{\tp}[1]{{P}_{#1}}

\newcommand{\E}{\mathcal{E}}
\newcommand{\tr}[2]{\mathrm{tr}_{#1}\left[ #2 \right]}
\newcommand{\bip}{\mathbb{C}^d\otimes\mathbb{C}^d}

\newcommand{\norm}[1]{\left\|#1\right\|}

\newcommand{\rIso}{\rho_{\rm iso}}
\newcommand{\id}{\mathbb{I}}

\newcommand{\F}[1]{\mathcal{F}_{\rm max}(#1)}

\newcommand{\PT}{\mathcal{P}}
\newtheorem{theorem}{Theorem}

\newtheorem{lemma}[theorem]{Lemma}
\newtheorem{proposition}[theorem]{Proposition}
\newtheorem{alemma}{Lemma}[section]

\newtheorem{adefinition}[alemma]{Definition}

\newtheorem{question}{Question}

\newtheorem{definition}{Definition}

\definecolor{nred}{rgb}{0.9,0.1,0.1}
\definecolor{nblack}{rgb}{0,0,0}
\definecolor{nblue}{rgb}{0.2,0.2,0.8}
\definecolor{ngreen}{rgb}{0.2,0.6,0.2}
\definecolor{ublue}{rgb}{0,0,0.5}
\definecolor{pur}{rgb}{0.75,0,0.75}
\definecolor{nngrn}{rgb}{0,0.5,0.5}

\newcommand{\blu}{\color{nblue}}

\begin{document}
\title{Entanglement preserving local thermalization}

\author{Chung-Yun Hsieh}
\email{chung-yun.hsieh@icfo.eu}
\affiliation{ICFO - Institut de Ci\`encies Fot\`oniques, The Barcelona Institute of Science and Technology, 08860 Castelldefels, Spain}

\author{Matteo Lostaglio}
\affiliation{ICFO - Institut de Ci\`encies Fot\`oniques, The Barcelona Institute of Science and Technology, 08860 Castelldefels, Spain}

\author{Antonio Ac\'in}
\affiliation{ICFO - Institut de Ci\`encies Fot\`oniques, The Barcelona Institute of Science and Technology, 08860 Castelldefels, Spain}\affiliation{ICREA-Instituci\'o Catalana de Recerca i Estudis Avan\c{c}ats, Lluis Companys 23, 08010 Barcelona, Spain}

\date{\today}

\begin{abstract}
We investigate whether entanglement can survive the thermalization of subsystems. We present two equivalent formulations of this problem: (1) Can two isolated agents, accessing only pre-shared randomness, locally thermalize arbitrary input states while maintaining some entanglement? (2) Can thermalization with local heat baths, which may be classically correlated but do not exchange information, locally thermalize arbitrary input states while maintaining some entanglement? We answer these questions in the positive at every nonzero temperature and provide bounds on the amount of preserved entanglement. We provide explicit protocols and discuss their thermodynamic interpretation: we suggest that the underlying mechanism is a speed-up of the subsystem thermalization process. We also present extensions to multipartite systems. Our findings show that entanglement can survive locally performed thermalization processes accessing only classical correlations as a resource. They also suggest a broader study of the channel's ability to preserve resources and of the compatibility between global and local dynamics.
\end{abstract}

\maketitle

\section{Introduction}\label{Sec:Introduction}
Entanglement is a core feature of quantum theory and one of the most representative resources in quantum information science. 
In fact, it is at the basis of quantum advantages in metrology \cite{Dobrzanski2014}, cryptography \cite{Pironio2010}, communication \cite{Cleve1997} and computation \cite{Bravyi2017}. 
Entanglement also impacts quantum thermodynamic protocols, e.g., by allowing one to extract more work than what would be possible with classical correlations \cite{Oppenheim2002,Perarnau-Llobet2015,Funo2013}, resulting in negative work cost of erasure~\cite{del_Rio2011} and strong heat backflows \cite{Jennings2010}. 
Entanglement is also a crucial ingredient to understand local equilibration~\cite{Kaufman2016} and its compatibility with global unitary evolution~\cite{popescu2006entanglement,Linden2009}.

While being a powerful resource, entanglement often does not survive interactions with an external environment. It is therefore a central question whether entanglement can be maintained by certain classes of dynamics.
From a thermodynamic point of view, one important class is thermalization, describing the evolution of generic states toward thermal equilibrium.
Formally, {\em thermalization} is defined as a transformation mapping {\em arbitrary} input states to a fixed output state -- the thermal state.
While entanglement is distributed at spatially separated locations, thermalization often acts locally and is known to destroy quantum correlations.
It is therefore important to know whether quantum theory allows global entanglement to persist after locally performed thermalizations.

One way to formalize this question is as follows.
Suppose an unknown input state is distributed to two agents at spatially separated locations.
We assume that the agents neither share additional quantum resources, such as another entangled state, nor can they communicate with each other. 
Each of them  has access to a local heat bath, and we allow for the two baths to be classically correlated across the bipartition.
Each party thermalizes their half of the (unknown) input state by coupling their local systems to the correspondent local bath. 
We assume the two dynamics remain independent, for example due to the timescales involved. 
Our central question is whether entanglement can survive when the local systems are thermalized [Fig.~\ref{Fig:TwoInterpretation} (a)].

We will show that the above question admits an {\em equivalent} reformulation as follows. 
Suppose two agents are restricted to perform local operations (LO) and can exploit pre shared randomness (SR) -- a set of physical dynamics (or simply {\em channels}~\cite{footnote-channel}) denoted by LOSR. 
The question above is then equivalently phrased as follows: Is there an LOSR channel that (i) locally thermalizes every input to predefined thermal states (i.e., it is locally indistinguishable from a thermalization) and (ii) the output is entangled at least for some input [Fig.~\ref{Fig:TwoInterpretation} (b)]?
Such channels, whose existence we want to explore, will be called {\em entanglement preserving local thermalizations} (EPLTs).
 
Ultimately, these are fundamental questions concerning the structure of quantum mechanics, specifically about the interplay between subsystem thermalization and quantum correlations. 
Here we ask if classical correlations/shared randomness alone allow for the preservation of entanglement in thermalizations. 
The answer that we find suggests that shared randomness can be a useful resource to sustain entanglement during thermalization.

\begin{figure}
\scalebox{0.8}{\includegraphics{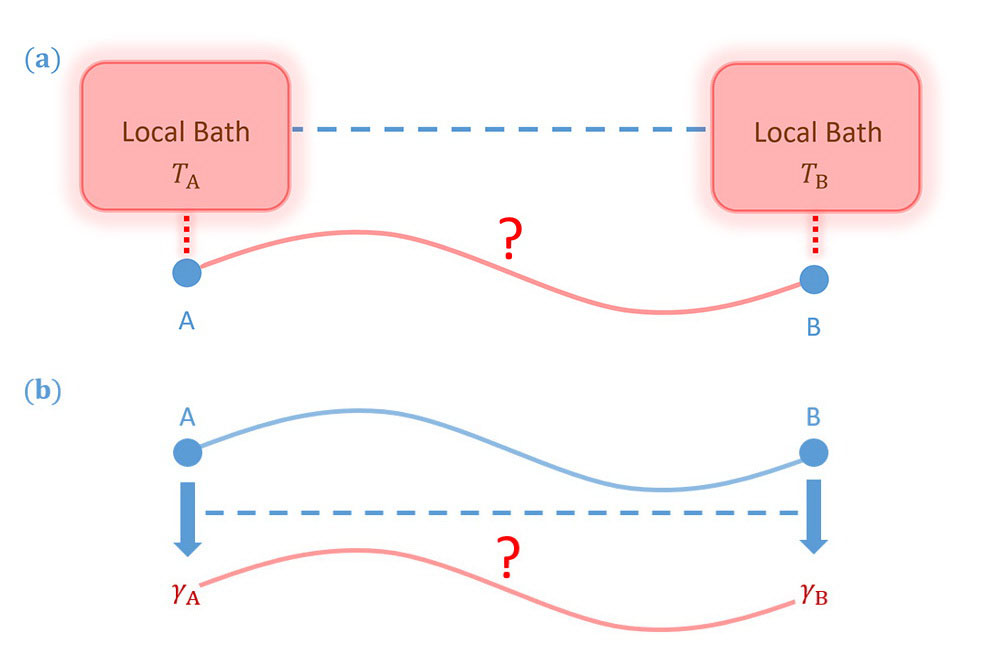} }
\caption{Schematic interpretation of the two formulations for the EPLT question. Dashed lines represent classical correlations, and continuous lines represent quantum correlations.
(a) \em Thermodynamic formulation. In this formulation, we ask whether entanglement can survive after subsystem thermalizations are achieved by coupling to classically correlated heat baths.
(b) \em Information-theoretic formulation. In this formulation, we ask whether entanglement can survive after a local operation plus shared randomness (LOSR) channel that is locally indistinguishable from a thermalization process.}
\label{Fig:TwoInterpretation} 
\end{figure}

\section{Main Question}
We first formalize the question described in the introduction in terms of a local thermalization task under LOSR channels, which are local dynamics assisted by pre shared classical correlations.
Formally, a bipartite LOSR channel is defined as 
\begin{eqnarray}
\mathcal{E} = \int (\mathcal{E}_{\rm A}^\lambda\otimes\mathcal{E}_{\rm B}^\lambda)p_\lambda d\lambda
\end{eqnarray} 
with $p_\lambda \geq 0$ and $\int d\lambda p_\lambda =1$, where $\mathcal{E}_{\rm A}^\lambda,\mathcal{E}_{\rm B}^\lambda$ are channels on the local systems for every $\lambda$.
To illustrate how to realize an LOSR channel, two local agents can share classical randomness, e.g., a third party samples $\lambda$ and, according to the probability distribution $p_\lambda$, distributes the outcomes to them, at the beginning of the task.
Each $\lambda$ corresponds to a specific local operation for each local agent (captured by $\mathcal{E}_{\rm A}^\lambda,\mathcal{E}_{\rm B}^\lambda$).
The local agents apply the corresponding local operations conditioned on the received $\lambda$.

Now, consider two spatially separated agents Alice (A) and Bob (B).
Each of them holds a system with local Hamiltonian $H_{\rm X}$ $({\rm X = A,B})$, and the total Hamiltonian is $H_{\rm A} \otimes \id + \id \otimes H_{\rm B}$.  
By means of a local process, they want to thermalize their local system to some local environment temperature $T_{\rm A}$ and $T_{\rm B}$, respectively. 
Denote the thermal state $\gamma_{\rm X}\coloneqq \frac{e^{-H_{\rm X}/kT_{\rm X}}}{{\rm tr}\left(e^{-H_{\rm X}/kT_{\rm X}}\right)}$, where $k$ is Boltzmann's constant.
We can now state the main definition:
\begin{definition}
A channel $\mathcal{E}$ on AB is a {\em local thermalization} to $(\gamma_{\rm A},\gamma_{\rm B})$ if
\begin{enumerate}
	\item \label{condition1} $\mathcal{E}$ is an LOSR channel;
	\item \label{condition2} $\tr{\rm A}{ \mathcal{E}(\rho_{\rm AB})} = \gamma_{\rm B}$, \; $\tr{\rm B}{ \mathcal{E}(\rho_{\rm AB})} = \gamma_{\rm A}$, \; $\forall\,\rho_{\rm AB}$.
\end{enumerate}
$\mathcal{E}$ is an \emph{entanglement preserving local thermalization} (EPLT) if, furthermore, there exists some input $\rho_{\rm AB}$ such that the output $\mathcal{E}(\rho_{\rm AB})$ is entangled. 
\end{definition}

In other words, a local thermalization $\mathcal{E}$ is {\em local} in two senses: that is, it is an LOSR channel (condition \ref{condition1}) that locally thermalizes every input (condition \ref{condition2}). 
An example of local thermalization is the channel \mbox{$\rho_{\rm AB} \mapsto \gamma_{\rm A} \otimes \gamma_{\rm B}$} for every $\rho_{\rm AB}$, which is not an EPLT. 

To gather intuition on the above definition, note that dropping either of the two conditions trivializes the dynamical question.
If we drop condition~\ref{condition1}, $\E$ can be any channel and, in particular, it can be any state preparation; then, our dynamical question concerning the existence of EPLTs is reduced to the ``static'' question of the existence of entangled states with given thermal marginals. 
Condition~\ref{condition1} avoids this trivialization by asking that $\E$ is an LOSR channel, which means that entanglement cannot be created~\cite{QCI-book}.

Also, note that if we drop the requirement that condition~\ref{condition2} holds for every state, then again the existence of entangled locally thermal states and the identity channel would trivially satisfy the requirements. 
For example, if $H_{\rm A} = H_{\rm B} = E \ket{E}\bra{E}$, and $T_{\rm A} = T_{\rm B}= T$, it would be enough to observe that the two qubit state $\sqrt{1/Z} \ket{00} + \sqrt{e^{-E/kT}/Z} \ket{11}$ (with $Z=1+e^{-E/kT}$) is locally thermal and entangled for every $T>0$. 
Hence, condition~\ref{condition2} crucially requires that the thermalization protocol does not depend on the input, which is also realistic from an operational point of view.

On the other hand, one may ask whether we could strengthen condition~\ref{condition1} by asking $\mathcal{E}$ to be a local operation {\em without} shared randomness.
However, as expected, no correlation, even classical ones, can be preserved in this scenario:
\begin{proposition}\label{Result:NoProduct}
Any product local thermalization to the marginals $(\gamma_{\rm A},\gamma_{\rm B})$ coincides with the constant channel \mbox{$(\cdot)\mapsto\gamma_{\rm A}\otimes\gamma_{\rm B}$.}
In other words, no correlation can be preserved by product local thermalizations.
\end{proposition}
\begin{proof}
Suppose $\mathcal{E}$ is a product local thermalization given by $\mathcal{E} = \mathcal{E}_{\rm A}\otimes\mathcal{E}_{\rm B}$.
By definition $\mathcal{E}_{\rm X}$ is identical to the constant channel $(\cdot)\mapsto\gamma_{\rm X}$,  which is a measure and prepare channel, thereby being an entanglement-breaking channel~\cite{Horodecki2003}.
This means $\mathcal{E}(\rho_{\rm AB}) = (\mathcal{E}_{\rm A}\otimes\mathcal{E}_{\rm B})(\rho_{\rm AB})$ is always a separable state.
Since $(\mathcal{E}_{\rm A}\otimes\mathcal{E}_{\rm B})\circ(\mathcal{E}_{\rm A}\otimes\mathcal{E}_{\rm B}) = \mathcal{E}_{\rm A}\otimes\mathcal{E}_{\rm B}$, for an arbitrary $\rho_{\rm AB}$, we have
\mbox{$
(\mathcal{E}_{\rm A}\otimes\mathcal{E}_{\rm B})(\rho_{\rm AB}) = \sum_i f_i \rho_{\rm A}^i\otimes\rho_{\rm B}^i
$}
for some $f_i \geq 0$, $\sum_i f_i =1$.
Then, for an arbitrary $\rho_{\rm AB}$,
\begin{eqnarray}
(\mathcal{E}_{\rm A}\otimes\mathcal{E}_{\rm B})(\rho_{\rm AB}) &&= (\mathcal{E}_{\rm A}\otimes\mathcal{E}_{\rm B})\circ(\mathcal{E}_{\rm A}\otimes\mathcal{E}_{\rm B})(\rho_{\rm AB}) \nonumber\\
&&=(\mathcal{E}_{\rm A}\otimes\mathcal{E}_{\rm B})(\sum_i f_i \rho_{\rm A}^i\otimes\rho_{\rm B}^i) \nonumber\\
&&= \sum_i f_i \gamma_{\rm A}\otimes\gamma_{\rm B} \nonumber\\
&&= \gamma_{\rm A}\otimes\gamma_{\rm B},
\end{eqnarray}
which completes the proof.
\end{proof}
Hence, the simplest EPLT, if it exists, must exploit shared randomness to preserve entanglement during a local thermalization process.

What we have formalized here is the information-theoretic formulation of the main problem [Fig.~\ref{Fig:TwoInterpretation} (b)].
An alternative formulation is the thermodynamic one [Fig.~\ref{Fig:TwoInterpretation} (a)], where LOSR is replaced by classically correlated heat baths, as mentioned in Sec.~\ref{Sec:Introduction}. 
The formal definition of the thermodynamic formulation and the proof of the equivalence between two formulations are given in Sec.~\ref{Sec:Eqv}. 
This equivalence allow us to rigorously analyze the problem in the information-theoretic formulation, while ensuring that we can always map back to a thermodynamic setting for a clearer physical meaning.

 Our main question is then as follows:
\begin{center}
{\em Do entanglement preserving local thermalizations exist?}
\end{center}

\section{Existence of EPLT}
We now turn to the existence of EPLT at every nonzero local temperature.
From now on we assume equal finite local dimension $d$, no degeneracies, and finite energies.
First, we need to introduce the {\em $(U\otimes U^*)$-twirling operation}~\cite{Horodecki1999,Bennett1996}, which is defined by
\begin{eqnarray}
\mathcal{T}(\rho_{\rm AB})\coloneqq\int_{U(d)}(U\otimes U^*)\rho_{\rm AB}(U\otimes U^*)^\dagger dU,
\end{eqnarray} 
where the integration is taken over the group $U(d)$ of unitary operators in dimension $d$ with Haar measure $dU$.
Operationally, twirling results from the application of coordinated random local unitaries.
This turns the local systems into maximally entropic states, while correlations between them can still exist after the operation.

The output of $\mathcal{T}$ is always an isotropic state~\cite{Horodecki1999}, 
\begin{eqnarray}
\rIso(p)\coloneqq p\proj{\Psi_d^+} + (1-p)\frac{\id_{\rm AB}}{d^2},
\end{eqnarray}
where \mbox{$\ket{\Psi_d^+}\coloneqq\frac{1}{\sqrt{d}}\sum_{n=0}^{d-1}\ket{n}\otimes\ket{n}$} is maximally entangled, and $p\in[-\frac{1}{d^2-1},1]$ due to the positivity of quantum states.
Furthermore, it can preserve entanglement since $\bra{\Psi_d^+}\mathcal{T}(\rho_{\rm AB})\ket{\Psi_d^+} = \bra{\Psi_d^+}\rho_{\rm AB}\ket{\Psi_d^+}$~\cite{Horodecki1999}, together with the fact that $\rIso$ is entangled if and only if \mbox{$ \bra{\Psi_d^+}\rIso\ket{\Psi_d^+}>\frac{1}{d}$}~\cite{Horodecki1999}. 
Hence $\mathcal{T}(\rho_{\rm AB})$ is entangled if and only if $\bra{\Psi_d^+}\rho_{\rm AB}\ket{\Psi_d^+}>\frac{1}{d}$.

Formally, we introduce the first candidate EPLT, 
\begin{eqnarray}\label{Eq:TheMap}
\mathcal{E}^{\epsilon}(\cdot) \coloneqq \mathcal{D}_{\eta_{\rm A}^\epsilon\otimes\eta_{\rm B}^\epsilon}^{(1-\epsilon)}\circ\mathcal{T}(\cdot),
\end{eqnarray}
\begin{equation}
\label{eq:mixingchannel}
\mathcal{D}_\sigma^p(\cdot)\coloneqq p\sigma + (1-p)\mathcal{I}(\cdot),
\end{equation}
where $\eta_{\rm X}^\epsilon\coloneqq \gamma_{\rm X} + \frac{\epsilon}{1-\epsilon} \left(\gamma_{\rm X} - \frac{\id_{\rm X}}{d}\right)$,
with $\epsilon\in[0,1]$ to be defined, and $\gamma_{\rm X}$ is the given local thermal state. 
In the information-theoretic formulation, this channel can be understood in two steps: First, local agents A and B realize the twirling operation $\mathcal{T}$, and then they use pre-shared randomness to mix the output of the twirling with the state $\eta_{\rm A}^\epsilon\otimes\eta_{\rm B}^\epsilon$, with the desired probability $1-\epsilon$.

Before stating the main result, first we show that the channels defined by Eq.~\eqref{Eq:TheMap} are local thermalization in certain parameter regimes.
\begin{lemma}\label{Lemma:LT}
$\mathcal{E}^\epsilon$ is a local thermalization to $(\gamma_{\rm A},\gamma_{\rm B})$ for all
$
0\le\epsilon\le\epsilon_*\coloneqq d\tp{\rm min},
$
where $\tp{\rm min}$ is the smallest eigenvalue among $\gamma_{\rm A}$ and $\gamma_{\rm B}$.
\end{lemma}
\begin{proof}
First, the definition of $\eta_{\rm X}^\epsilon$ implies that $\mathcal{E}^\epsilon$ will locally behave as a thermalization.
More precisely, Eq.~\eqref{Eq:TheMap} and the definition of $\eta_{\rm X}^\epsilon$ give
\begin{eqnarray}
{\rm tr}_{\rm A}\left[\mathcal{E}^\epsilon(\cdot)\right] =(1-\epsilon)\eta_{\rm B}^\epsilon + \epsilon\frac{\id_{B}}{d}= \gamma_{\rm A},
\end{eqnarray}
and the same by exchanging A and B.
Hence, it suffices to show that $\mathcal{E}^\epsilon$  is an LOSR channel in order to prove it is a local thermalization.
From the definition, it suffices to prove that $\eta_{\rm X}^\epsilon$ is a quantum state when $\epsilon$ falls into the prescribed region.
Write $\gamma_{\rm X} = \sum_{n=0}^{d-1}P_n^{\rm X}\proj{n}_{\rm X}$ with $1\ge P_0^{\rm X}\ge P_1^{\rm X}\ge ...\ge P_{d-1}^{\rm X}\ge0$.
From the definition $\eta_{\rm X}^\epsilon\coloneqq \gamma_{\rm X} + \frac{\epsilon}{1-\epsilon} \left(\gamma_{\rm X} - \frac{\id_{\rm X}}{d}\right)$, we have $\eta_{\rm X}^\epsilon=\sum_{n=0}^{d-1}Q_n^{\rm X}\proj{n}$, with 
\begin{eqnarray}\label{Eq:QX}
Q^{\rm X}_{n} = \frac{1}{1-\epsilon}P^{\rm X}_{n} - \frac{\epsilon}{d(1-\epsilon)}.
\end{eqnarray}
Since $\epsilon \leq 1$, we have the hierarchy $Q^{\rm X}_{0}\ge Q^{\rm X}_{1}\ge...\ge Q^{\rm X}_{d-1}$ and the normalization condition $\sum_{n=0}^{d-1}Q^{\rm X}_{n} = 1$.
Hence, it suffices to impose $ Q^{\rm X}_{d-1} \geq 0$ to make sure $\eta_{\rm X}^\epsilon$ is a quantum state. 
This gives $\epsilon \le d P^{\rm X}_{d-1}$ (we have $P^{\rm \rm X}_{d-1} \leq \frac{1}{d}$, since $\gamma_{\rm X}$ is a thermal state) for ${\rm X=A,B}$, which leads to the desired range $0\le\epsilon\le\epsilon_*\coloneqq d\tp{\rm min}$.
\end{proof}
With the above lemma in hand, our first main result can be stated as follows (note again that we only consider finite-energy Hamiltonians):
\begin{theorem}\label{Theorem:Existence}
$\mathcal{E}^{\epsilon_*}$ is an EPLT to $(\gamma_{\rm A},\gamma_{\rm B})$ for all \mbox{$T_{\rm A},T_{\rm B}>0$}.
\end{theorem}
\begin{proof}
It suffices to show that the output state is entangled when the input state is $\ket{\Psi_d^+}$.
We use the {\em positive partial transpose} (PPT) criterion~\cite{Peres1996, Horodecki1996} to detect the entanglement of the output.
Again, write $\eta_{\rm X}^\epsilon=\sum_{n=0}^{d-1}Q_n^{\rm X}\proj{n}$; then we have
\begin{eqnarray}
&&\mathcal{E}^{\epsilon_*}(\proj{\Psi_d^+})\nonumber\\
&&= \epsilon_*\proj{\Psi_d^+} + (1-\epsilon_*)\sum_{n,m=0}^{d-1}Q_n^{\rm A}Q_m^{\rm B}\proj{nm}\nonumber\\
&&=\sum_{n,m=0}^{d-1}\left[\frac{\epsilon_*}{d}\ket{nn}\bra{mm} + (1-\epsilon_*)Q_n^{\rm A}Q_m^{\rm B}\proj{nm}\right].\quad\;
\end{eqnarray}
Now we take the partial transpose on the subsystem ${\rm B}$, which gives the following operator:
\begin{eqnarray}
&&\sum_{n,m=0}^{d-1}\left[\frac{\epsilon_*}{d}\ket{nm}\bra{mn} + (1-\epsilon_*)Q_n^{\rm A}Q_m^{\rm B}\proj{nm}\right]\nonumber\\
&&=\left(\bigoplus_{n\neq m}M_{nm}\right)\oplus D,
\end{eqnarray}
where 
\begin{eqnarray}
M_{nm}\coloneqq \left(\begin{matrix}(1-\epsilon_*)Q_n^{\rm A}Q_m^{\rm B} & \frac{\epsilon_*}{d}\\ \frac{\epsilon_*}{d} & (1-\epsilon_*)Q_m^{\rm A}Q_n^{\rm B}\end{matrix}\right),
\end{eqnarray}
and
\begin{eqnarray}
D\coloneqq\bigoplus_{n=0}^{d-1} \left[\frac{\epsilon_*}{d}+(1-\epsilon_*)Q_n^{\rm A}Q_n^{\rm B}\right]
\end{eqnarray}
is the contribution of the diagonal terms.
To see that the output is entangled, it suffices to show that there exists a negative eigenvalue of at least one $M_{nm}$.
To this end, we first note that  when we substitute $\epsilon = \epsilon_* = dP_{\rm min}$ in Eq.~\eqref{Eq:QX}, we have $Q_{d-1}^{\rm A} = 0$ (without loss of generality, we assume $P_{\rm min} = P^{\rm A}_{d-1}$).
This means that for every $m<d-1$, we have
\begin{eqnarray}
M_{d-1,m} = \left(\begin{matrix}0 & \frac{\epsilon_*}{d}\\ \frac{\epsilon_*}{d} & (1-\epsilon_*)Q_m^{\rm A}Q_{d-1}^{\rm B}\end{matrix}\right),
\end{eqnarray}
which will have a negative eigenvalue if the off-diagonal terms are positive; namely, when $\epsilon_*>0$.
This completes the proof.
\end{proof}
Theorem~\ref{Theorem:Existence} shows the existence of EPLT in the most general case for bipartite systems, apart from the special case of zero temperature.
We leave the discussion of the zero-temperature case and further remarks for the following independent sections.
We note that as $T_{\rm A},T_{\rm B} \rightarrow + \infty$, we have $\epsilon_* \rightarrow 1$ and $\mathcal{E}^{\epsilon_*} \rightarrow \mathcal{T}$ (this also means twirling operation is an EPLT with infinite local temperatures or fully degenerate local Hamiltonians). 
In this sense, $\mathcal{E}^{\epsilon}$ can be considered as a finite temperature extension of the twirling operation.

\subsection{Zero-Temperature Case Discussion}
The previous theorem leaves out only the case $T_{\rm A}= T_{\rm B}=0$, which we treat separately here.
Again assuming equal finite local dimension $d$, we separately consider two cases, depending on whether or not there is ground-state degeneracy on the systems.
In the latter case, the corresponding local thermal state will be given by the unique {\em pure} ground state of the local Hamiltonian. Then one can immediately conclude that no entanglement can be preserved, because a pure state cannot be correlated with any other system.
Hence, no EPLT exists. 
On the other hand, if both local systems admit ground-state degeneracy, then EPLTs exist even when $T_{\rm A} = T_{\rm B} = 0$.
This can be realized by means of an energy measurement, a standard thermalization of the local systems conditioned on the measurement results, and a global $(U\otimes U^*)$-twirling operation within the ground energy subspace.
We leave the details to Appendix~\ref{App:Degenerate}.
This confirms that EPLTs exist in all nontrivial scenarios.

\subsection{Improved Entanglement Certification of Qubit-Qubit Systems}
As we discussed, EPLTs exist for all nonzero temperatures and finite-energy Hamiltonians by using the sufficiency of the PPT criterion. 
In fact, since it is also a necessary condition in two-qubit systems, a more detailed characterization of the output entanglement can be given in this case:
\begin{proposition}\label{Prop:2qubit}
For a two qubit system with 
\mbox{$\gamma_{\rm A} = \gamma_{\rm B}\neq\proj{0}$}, one has that $\mathcal{E}^{\epsilon_*}(\rho_{\rm AB})$ is entangled if and only if $\bra{\Psi_2^+}\rho_{\rm AB}\ket{\Psi_2^+} > \frac{1}{2}$. 
\end{proposition}
\begin{proof}	
Setting $\gamma_{\rm X} = (1-q) \ket{0}\bra{0}_{\rm X} + q \ket{1}\bra{1}_{\rm X}$, a direct computation shows $\eta_{\rm A}^{\epsilon_*} \otimes \eta_{\rm B}^{\epsilon_*} = \ket{00}\bra{00}$ and, if \mbox{$p = \frac{ 4\bra{\Psi_2^+}\rho_{\rm AB}\ket{\Psi_2^+}-1}{3}$}, 
\begin{eqnarray}
&\mathcal{E}^{\epsilon_*}\left(\rho_{\rm AB}\right) &= (1-2q) \eta_{\rm A}^{\epsilon_*}\otimes\eta_{\rm B}^{\epsilon_*} + 2q \rIso(p)\nonumber\\
&&=A\proj{11}+B\proj{00} \nonumber\\
&&+ C(\proj{10}+\proj{01})\nonumber\\
&&+D\left(\ket{11}\bra{00} + \ket{00}\bra{11}\right),
\end{eqnarray}
where
\begin{eqnarray}
&&A= q\times \frac{1+p}{2}, B=(1-2q) + q\times \frac{1+p}{2},\nonumber\\
&&C = q\times\frac{1-p}{2}, D=qp.
\end{eqnarray}
We note that $A,B,C,D$ are all non-negative.
This means the partial transpose on Bob's side has a negative eigenvalue if and only if $C - |D| < 0$, which gives $p>\frac{1}{3}$, or, equivalently, $\bra{\Psi_2^+}\mathcal{T}(\rho_{\rm AB})\ket{\Psi_2^+} > \frac{1}{2}$. Since \mbox{$\bra{\Psi_2^+}\mathcal{T}(\rho_{\rm AB})\ket{\Psi_2^+} =  \bra{\Psi_2^+}\rho_{\rm AB}\ket{\Psi_2^+}$}, the result follows by using the PPT criterion.
\end{proof}

\subsection{Bounds on Preserved Entanglement of EPLT}
Since we proved that EPLTs universally exist, a natural question is whether it is possible to quantify the output entanglement.
For the sake of measuring quantum correlations at the end of the local thermalization, we consider the {\em fully entangled fraction }(FEF)~\cite{Horodecki1999-2,Albeverio2002}. 
For a given bipartite quantum state $\rho_{\rm AB}$ acting on $\bip$, FEF is defined by
\begin{eqnarray}
\F{\rho_{\rm AB}}\coloneqq\max_{\ket{\Phi_d}}\bra{\Phi_d}\rho_{\rm AB}\ket{\Phi_d},
\end{eqnarray}
where the optimization is taken over all maximally entangled states $\ket{\Phi_d}\in\bip$.
A well-known fact about FEF is its capacity to characterize different entanglement and nonlocal properties~\cite{Horodecki1999-2,Albeverio2002, RMP-Bell, Hsieh2016,Hsieh2018E,Zhao2010,Cavalcanti2013,Quintino2016,RMP-Ent}. 
Then a direct computation shows the following bound for the output entanglement for the channel given in Eq.~\eqref{Eq:TheMap}:

\begin{proposition}
For all input states $\rho_{\rm AB}$, we have
\begin{eqnarray}
\label{eq:fmaxbound}
\mathcal{F}_{\rm max}[\mathcal{E}^\epsilon(\rho_{\rm AB})]\ge \epsilon\bra{\Psi_d^+}\rho_{\rm AB}\ket{\Psi_d^+}.
\end{eqnarray}
\end{proposition}
\begin{proof}
Recall $\bra{\Psi_d^+}\mathcal{T}(\rho_{\rm AB})\ket{\Psi_d^+} = \bra{\Psi_d^+}\rho_{\rm AB}\ket{\Psi_d^+}$~\cite{Horodecki1999}, which implies 
\begin{eqnarray}
\mathcal{F}_{\rm max}[\mathcal{E}^\epsilon(\rho_{\rm AB})] &&\ge(1-\epsilon)\bra{\Psi_d^+}(\eta_{\rm A}^\epsilon\otimes\eta_{\rm B}^\epsilon)\ket{\Psi_d^+} \nonumber\\
&&+ \epsilon\bra{\Psi_d^+}\rho_{\rm AB}\ket{\Psi_d^+}.
\end{eqnarray}
Since $\bra{\Psi_d^+}(\eta_{\rm A}^\epsilon\otimes\eta_{\rm B}^\epsilon)\ket{\Psi_d^+}\ge0$, the proof is completed.
\end{proof}
By taking $\epsilon=\epsilon_*=d\tp{\rm min}$, a sufficient condition for the output of $\mathcal{E}^{\epsilon_*}$ to be entangled is then 
\begin{eqnarray}
\tp{\rm min}> \frac{1}{\bra{\Psi_d^+}\rho_{\rm AB}\ket{\Psi_d^+} d^2 }.
\end{eqnarray}

\subsection{Multipartite Extension}
The existence of EPLT in the multipartite case can be established by using a multipartite twirling and the corresponding entanglement fraction. 
In particular, it can be shown that genuine multipartite entanglement of the {\em Greenberger-Horne-Zeilinger} (GHZ) state~\cite{GHZ} can be preserved by local thermalizations. 
We leave the detailed analysis to Appendix~\ref{App:Multipartite}.

\section{Equivalence Between Information-Theoretic And Thermodynamic Formulations}\label{Sec:Eqv}
So far we have analyzed the information-theoretic formulation [Fig.~\ref{Fig:TwoInterpretation} (b)] and we can now formalize the thermodynamic formulation [Fig.~\ref{Fig:TwoInterpretation} (a)].
Schematically, this formulation depicts two local systems interacting with their individual heat baths and thermalizing.
The heat baths do not interact between them, but can share some initial classical correlations.
The resulting dynamics can hence be characterized as follows:

\begin{definition}\label{Def:Bath}
A channel $\mathcal{C}$ of a bipartite system {\rm AB} is a {\em local bath thermalization} to $(\gamma_{\rm A},\gamma_{\rm B})$ if
\begin{enumerate}
\item $\mathcal{C}(\rho_{\rm AB}) = \tr{\rm A'B'}{\mathcal{V}_{\rm AA'} \otimes \mathcal{V}_{\rm BB'} (\rho_{\rm AB} \otimes \gamma_{\rm A'B'}) }$, where $\mathcal{V}_{\rm XX'}(\cdot)\coloneqq U_{\rm XX'}(\cdot)U_{\rm XX'}^\dagger$ are local unitary dynamics on ${\rm XX'}$ and $\gamma_{\rm A'B'}$ is a separable thermal state.  
\item $\tr{\rm A}{ \mathcal{C}(\rho_{\rm AB})} = \gamma_{\rm B}$, \; $\tr{\rm B}{ \mathcal{C}(\rho_{\rm AB})} = \gamma_{\rm A}$, \; $\forall \rho_{\rm AB}$.
\end{enumerate}
$\mathcal{C}$ is an {\em entanglement preserving local bath thermalization} if there exists $\rho_{\rm AB}$ such that $\mathcal{C}(\rho_{\rm AB})$ is entangled.
\end{definition}

The above notion illustrates the thermodynamic formulation, and the alternative form of the question in Sec.~\ref{Sec:Introduction} is then:
{\em Do entanglement preserving local bath thermalizations exist?}
The following result allows us to rephrase the results in this new formulation:
\begin{theorem}\label{Prop:GeneralEquivalence}
\label{prop:equivalence}
A bipartite channel is a local bath thermalization if and only if it is a local thermalization. 
\end{theorem}
The proof is given in Appendix~\ref{App:Eqv}.
Theorem~\ref{prop:equivalence} has the following consequence: If two local agents perform local interactions with a thermal bath that thermalize their local state for every input, even knowing that the bath has no entanglement across the bipartition, they still cannot conclude that their output is separable. 
Classical correlations alone in the bath can allow for the preservation of entanglement in the system, even after locally the thermalization is complete: entanglement preserving local bath thermalizations exist, for some separable (non product) $\gamma_{\rm A'B'}$.
We note that Proposition~\ref{Result:NoProduct} and Theorem~\ref{prop:equivalence}  imply that no entanglement preserving local bath thermalization exists if we restrict \mbox{$\gamma_{\rm A'B'} = \gamma_{\rm A'} \otimes \gamma_{\rm B'}$} in the setting of Theorem~\ref{prop:equivalence}. 

As an implication, Theorems~\ref{Theorem:Existence} and~\ref{Prop:GeneralEquivalence} imply that classical correlations in the bath are sufficient to preserve quantum correlations in the system for some input states, even after full thermalizations of the subsystems, at every nonzero temperature.

\section{Implementation and Mechanism}
Given the existence of EPLT, a natural and important question is the following:
\begin{center}
{\em
What is the mechanism behind EPLT?
}
\end{center}
To answer this question, we introduce another family of EPLT and use it to study the underlying physical reason.

\subsection{Alternative EPLT}
As the first step, we want to introduce an explicit thermodynamic protocol to achieve EPLTs.
To this end, we consider the following alternative EPLT construction:
\begin{eqnarray}\label{Eq:AlternativeEPLT}
\pe^{(\epsilon_{\rm A},\epsilon_{\rm B})}\coloneqq \left[\mathcal{D}_{\eta_{\rm A}^{\epsilon_{\rm A}}}^{(1-\epsilon_{\rm A})}\otimes\mathcal{D}_{\eta_{\rm B}^{\epsilon_{\rm B}}}^{(1-\epsilon_{\rm B})}\right]\circ\mathcal{T},
\end{eqnarray} 
with $\mathcal{D}_\sigma^p$ given in Eq.~\eqref{eq:mixingchannel}. 
The same proof shows that $\pe^{(\epsilon_{\rm A},\epsilon_{\rm B})}$ is a local thermalization to $(\gamma_{\rm A},\gamma_{\rm B})$ for all $0\le\epsilon_{\rm X}\le dP_{\rm min}^{\rm X}$, where $P_{\rm min}^{\rm X}$ is the smallest eigenvalue of $\gamma_{\rm X}$.
Moreover, we have \mbox{$\mathcal{F}_{\rm max}\left[\pe^{(\epsilon_{\rm A},\epsilon_{\rm B})}(\rho_{\rm AB})\right]\ge \epsilon_{\rm A}\epsilon_{\rm B}\bra{\Psi_d^+}\rho_{\rm AB}\ket{\Psi_d^+}$}.
This estimate means that the family $\pe^{(\epsilon_{\rm A},\epsilon_{\rm B})}$ includes EPLT when the local temperatures are not too low.

Eq.~\eqref{Eq:AlternativeEPLT} has a clear thermodynamic interpretation.
First, Alice and Bob perform the twirling (by applying random unitaries using pre-shared randomness). 
Then, they perform a sudden quench of the local system Hamiltonians $H_{\rm X} \mapsto H^{\epsilon_{\rm X}}_{\rm X}$ (where the energies are tuned, but not the eigenstates), with $\eta_{\rm X}^{\epsilon_{\rm X}} \propto e^{-H^{\epsilon_{\rm X}}_{\rm X}/kT_{\rm X} }$. 
At this point, by thermal contact with their local environments, they let their local system undergo 
a {\em partial thermalization}~\cite{footnote2}, 
\begin{eqnarray}\label{Eq:PTM}
\PT_\gamma^t(\cdot)\coloneqq e^{-\frac{t}{\tau_{\gamma}}}(\cdot) + \left(1-e^{-\frac{t}{\tau_{\gamma}}}\right)\gamma,
\end{eqnarray}
where $\tau_{\gamma}\in(0,\infty)$ is the thermalization time scale corresponding to $\gamma$.
Note that $\PT_\gamma^t=\mathcal{D}_\gamma^p$ with $p = 1- e^{\frac{-t}{\tau_\gamma}}$, and hence $\mathcal{D}_\gamma^p$ can be realized by a partial thermalization.
Finally, they quench their Hamiltonians back to $H_{\rm X}$. 
At this point, whatever the input was, the local states are $\gamma_{\rm X}$, i.e., A and B both have thermalized. 
However, quantum correlations can be preserved once local thermality is reached. This is in contrast to what happens if they each let their local system thermalize to an independent bath according to Eq.~\eqref{Eq:PTM}: in this case, local thermality is only reached when the global state is $\gamma_{\rm A} \otimes \gamma_{\rm B}$~\cite{footnote1}.

\subsection{Mechanism}
Since, in simple thermalization models, local thermality is reached only once correlations between the two parties are destroyed, the existence of EPLT suggests that the corresponding protocols rely on a  ``local speed-up'' of the thermalization. 
We will gather evidence for this intuition by taking the EPLT of Eq.~\eqref{Eq:AlternativeEPLT} as a model (with $\mathcal{D}_\gamma^p$ described by $\PT_\gamma^t$), showing that the local thermalization process is sped up through an LOSR channel that is able to preserve some entanglement. 
This makes sure that at local equilibrium, not all the (quantum) correlations are lost.

From Eq.~\eqref{Eq:PTM}, we learn that partial thermalization takes infinite time to thermalize the subsystem to $\gamma_{\rm X}$.
On the other hand, the local behavior of Eq.~\eqref{Eq:AlternativeEPLT} on the subsystem ${\rm X}$ 
is a random unitary followed by an incomplete partial thermalization with completion time \mbox{$t = -\tau_{\eta_{\rm X}^{\epsilon_{\rm X}}}\ln{\epsilon_{\rm X}}$}.
Since the completion time $t$ is always finite~\cite{footnote4}, we conclude that the subsystem thermalization is faster in the EPLT scheme when the time $t_{\mathcal{T}}$ to implement random unitaries is finite.

One may, however, suspect that in practice, the exact twirling requires $t_{\mathcal{T}}=\infty$; let us show that even if that is the case, the same speed-up argument holds. In fact, with a finite number $N$ of unitaries one can realize an approximation $\mathcal{T}^{(N)}$ of $\mathcal{T}$, with exponentially good precision in $N$~\cite{Toth2007,footnote3}. 
Since the completion time is $t_{\mathcal{T}^{(N)}} = N t_U$, with $t_U$ the time necessary to perform a single unitary, we will have $t_{\mathcal{T}^{(N)}}$ to scale logarithmically with the required precision $\delta$, with a constant prefactor $t_U$. 
This implies $t_{\mathcal{T}^{(N)}} \rightarrow \infty$   as $\delta \rightarrow 0$. 
However, as long as $t_U$ is sufficiently small compared with the typical thermalization time $\tau_{\gamma_{\rm X}} $, $\delta>0$ is small enough and $N$ is large enough, we expect a shorter time in the EPLT thermalization scheme compared with standard thermalization described by Eq.~\eqref{Eq:PTM}.
More precisely, we show that for any $\rho_{\rm X}\neq\gamma_{\rm X}$ and $\delta>0$ small enough, the EPLT scheme realizes a speed-up of $\delta$-thermalization~\cite{footnote6} with probability \mbox{$1 - O(\delta^4)$} whenever
\begin{eqnarray}\label{Eq:Speed-Up}
\tau_{\gamma_{\rm X}} > t_U\times\frac{8}{\ln2}.
\end{eqnarray}
We refer to Appendix~\ref{App:t_u} for the detailed proof and the complete statement of the theorem. 
In practice, the thermalization time-scale $\tau_{\gamma_{\rm X}}$ is often much longer than the time-scale $t_U$ of applying a single unitary operator, and hence the condition of Eq.~\eqref{Eq:Speed-Up} holds in various physical settings.

We finish the discussion by providing an example.
Suppose $\tau_{\eta_{\rm X}} = \tau_{\gamma_{\rm X}} = 100t_U$, which means it is possible to establish speed-up.
If one sets $P_{\rm min} = \frac{2}{d^2}$ and $\delta = 10^{-3}$, then we have $N = 92$ [from Eq.~\eqref{Eq:N_delta}].
This means we have speed-up for all $\rho_{\rm X}\neq\gamma_{\rm X}$ satisfying 
$
\norm{\rho_{\rm X}-\gamma_{\rm X}}_\infty >  d\times 0.00126
$
with success probability of implementation higher than $1 - 10^{-13}$.

\section{Conclusions}
We studied the robustness of quantum correlations under local thermalizations, which are a subset of local operations and classical communication (LOCC) that locally act as standard thermalizations. 
The main result can be summarized as showing that entanglement can survive under locally performed thermalizations at every nonzero temperature. 

This can be understood in two ways: on the one hand, it suggests that in the presence of local environments that degrade and eventually destroy shared quantum resources, in principle one could partially counter this detrimental effect by {\em actively} exploiting shared randomness. 
On the other hand, the thermodynamic formulation that we presented implies that several entangled particles thermalizing with a global bath do not necessarily end in a non-entangled state. 
This result may not seem surprising if the global bath is entangled, but is less clear if the bath displays only classical correlations, which is always going to be the case for high temperatures. 
In fact, our results imply that this process is possible at all possible nonzero temperatures.

We also investigated the mechanism behind the existence of EPLTs and we suggested that it can be traced back to a speed-up of the subsystem thermalization. 
We gave evidence here that protocols locally realizing fast preparations of thermal states may be exploited in conjunction with shared randomness to preserve global entanglement during thermalization processes.

From a foundational perspective, our work contributes to the research line that tries to identify genuinely quantum effects in a thermodynamic setting, by showing that a crucial ingredient of the quantum world can survive local thermalizations, and explores the relation between local and global thermalizations. 
For example, the existence of EPLTs implies that even if a local agent has witnessed a {\em local} thermalization of {\em every} input state, she could still subsequently witness a stronger than classical heat back-flow from the cold to the hot body due to the residual entanglement~\cite{Jennings2010}.
Our results open up new research directions. Within the resource theory of channels~\cite{Rosset2018,LiuWinter2019,LiuYuan2019,Theurer2019}, they suggest a general study of free operations' ability to preserve a given resource~\cite{Hsieh2019}.
Furthermore, EPLTs can be understood as an example of a ``quantum channel marginal problem,'' i.e., a dynamical version of the well-known quantum state marginal problem. Our work shows the compatibility between local preparations of given (full-rank) states and global LOCC channels that are not entanglement destroying. 
Given the importance of the state marginal problem~\cite{SchillingPhD}, we expect its dynamical version is also worth exploring.

\section*{Acknowledgements}
We thank (in alphabetical order) Flavio Baccari, Manabendra Bera, Joseph Bowles, Daniel Cavalcanti, Christian Gogolin, Alejandro Pozas Kerstjens, Ray-Kuang Lee, Yeong-Cherng Liang, Jan Ko{\l}ody\'nski, Mohammad Mehboudi, and Marco T\'ulio Quintino for fruitful discussions and useful comments.
The authors acknowledge support from the ICFOstepstone - PhD Programme for Early-Stage Researchers in Photonics funded by the Marie Sk{\l}odowska-Curie Co-funding of regional, national and international programmes (GA665884) of the European Commission, the European Union's Marie Sk{\l}odowska-Curie individual Fellowships (H2020-MSCA-IF-2017, GA794842), as well as the ERC AdG CERQUTE, the AXA Chair in Quantum Information Science, Spanish MINECO (QIBEQI FIS2016-80773-P and Severo Ochoa SEV-2015-0522), Fundaci\'o Privada Cellex, and the Generalitat de Catalunya (CERCA Program and SGR1381).

\onecolumngrid
\appendix

\section{Zero Temperature EPLT With Ground State Degeneracy}\label{App:Degenerate}
The idea is to perform the $(U\otimes U^*)$-twirling operation in the zero energy subspace.
To be precise, consider the following protocol, where we assume two-fold ground state degeneracy on both local systems to illustrate the idea.

\emph{Step 1:} On the local system ${\rm X}$, consider the projective measurement 
given by \mbox{$\{\Pi^{\rm X}_0,\id_{\rm X}-\Pi^{\rm X}_0\}$}, where $\Pi^{\rm X}_0$ is the projector onto the ground energy subspace:
\begin{eqnarray}
\Pi^{\rm X}_0\coloneqq \sum_{g=0,1} \ket{0,g}\bra{0,g}_{\rm X} ,
\end{eqnarray}
where $g$ is a degeneracy index and $\{\ket{0,g}\}_{g=0,1}$ span the ground energy subspace of the local Hamiltonian $H_{\rm X}$.
The first step of the protocol is to measure $\{\Pi^{\rm A}_0,\id_{\rm A}-\Pi^{\rm A}_0\}\otimes\{\Pi^{\rm B}_0,\id_{\rm B}-\Pi^{\rm B}_0\}$.
For each local agent, if the outcome reads $\Pi^{\rm X}_0$, nothing is done; if the outcome reads $\id_{\rm X} - \Pi^{\rm X}_0$, then the agent discards the original input and prepares $\frac{\Pi^{\rm X}_0}{2}$.

\emph{Step 2:} Use shared randomness to achieve a $(U\otimes U^*)$-twirling operation on the ground energy subspace, denoted by $\mathcal{T}^0$. Formally, the channel corresponding to the above protocol is $\mathcal{T}^0\circ(\mathcal{L}_{\rm A}\otimes\mathcal{L}_{\rm B})$, where
\begin{eqnarray}
\mathcal{L}_{\rm X}(\cdot)\coloneqq\Pi^{\rm X}_0(\cdot)\Pi^{\rm X}_0 + \Phi_{\frac{\Pi^{\rm X}_0}{2}}[(\id_{\rm X}-\Pi^{\rm X}_0)(\cdot)(\id_{\rm X}-\Pi^{\rm X}_0)]
\end{eqnarray}
 for ${\rm X=A,B}$, where $\Phi_{\rho}(\cdot)\equiv\rho$ is the channel discarding the input and preparing $\rho$.
 
Note that this protocol gives a local thermalization because the output states will have, independently of the input, marginal $\frac{\Pi^{\rm X}_0}{2}$ on the local system ${\rm X}$, which is the desired thermal state in this case.
The entanglement preservation can be seen by choosing the input state as $\frac{1}{\sqrt{2}}(\ket{0,0}_{\rm A}\otimes\ket{0,0}_{\rm B} + \ket{0,1}_{\rm A}\otimes\ket{0,1}_{\rm B})$, which is invariant under the whole protocol. 
This proves the existence of an EPLT.

\section{Multipartite EPLT}\label{App:Multipartite}
First, the definition of local thermalization can be generalized naturally: 
\begin{adefinition}
Consider a multipartite system $\bigotimes_{i=1}^N\mathcal{H}_i$.
For a given collection of $N$ single party thermal states $\{\gamma_i\}_{i=1}^N$, a channel $\mathcal{E}$ on $\bigotimes_{i=1}^N\mathcal{H}_i$ is called a {\em local thermalization} to $\{\gamma_i\}_{i=1}^N$ if 
\begin{enumerate}
\item It is of the form \mbox{$\mathcal{E} = \int_\lambda \bigotimes_{i=1}^N\mathcal{E}^\lambda_i p_\lambda d\lambda$}, where each $\mathcal{E}_i^\lambda$ is a channel on the $i^{th}$ local system, $p_\lambda \geq 0$ and $\int p_\lambda d\lambda =1$;
\item ${\rm tr}_{\backslash i}[\mathcal{E}(\rho)] = \gamma_i$ for every $\rho$ and $i$, where ${\rm tr}_{\backslash i}$ denotes trace over all but the $i^{th}$ system.
\end{enumerate} 
\end{adefinition}

A natural question is whether there exists {\em genuinely} multipartite EPLT; that is, a multipartite EPLT whose output is genuinely multipartite entangled for some input. We now show such channel exists. 

To do so, we consider an $N$-qubit system. 
 Using an appropriate sequence of $N$-local operations and shared randomness (see Sec.~IV B in Ref.~\cite{Dur2000}), one can define a channel, denoted by $\mathcal{T}_{\rm GHZ}$, which brings arbitrary $N$-qubit input states $\rho$ to the following form:

\begin{eqnarray}\label{Eq:GHZ-Twirling}
\sum_{\sigma = \pm}\lambda_0^\sigma\proj{\Psi^\sigma_0} + \sum_{j=1}^{2^{(N-1)} - 1}\lambda_j\left(\proj{\Psi^+_j} + \proj{\Psi^-_j}\right),
\end{eqnarray}
where
\begin{eqnarray}
\ket{\Psi_j^\pm}\coloneqq\frac{1}{\sqrt{2}}\left(\ket{j}\otimes\ket{0} \pm \ket{{2^{(N-1)} - j - 1}}\otimes\ket{1}\right),
\end{eqnarray}
 and binary notation is used ($j=j_1 \dots j_{N-1}$).
 We note that $\ket{\Psi_0^+}$ is the GHZ state $\frac{1}{\sqrt{2}} (\ket{0 \dots 0} + \ket{1 \dots 1})$~\cite{GHZ}.
 One important feature of $\mathcal{T}_{\rm GHZ}$ is the following preservation property: for all $j=0,...,2^{(N-1)} - 1$ and $\sigma = \pm$,
 \begin{eqnarray}
 \bra{\Psi_j^\sigma}\mathcal{T}_{\rm GHZ}(\rho)\ket{\Psi_j^\sigma}=\bra{\Psi_j^\sigma}\rho\ket{\Psi_j^\sigma}.
 \end{eqnarray}
 In particular, a state of the following form will be invariant under $\mathcal{T}_{\rm GHZ}$:
 \begin{eqnarray}
 x\proj{\Psi_0^+} + (1-x)\frac{\id}{2^N},
 \end{eqnarray}
and it is genuinely multipartite entangled if and only if $x>\frac{1}{1+2^{N-1}}$~\cite{Dur2000}.
Due to this fact, one can define the following map for a given set of $N$ single-qubit thermal states $\{\gamma_i\}_{i=1}^{N}$ extending Eq.~\eqref{Eq:TheMap}:
\begin{eqnarray}
\mathcal{E}^{(\epsilon,N)}(\cdot)\coloneqq(1-\epsilon)\bigotimes_{i=1}^{N}\eta_i + \epsilon\mathcal{T}_{\rm GHZ}(\cdot),
\end{eqnarray}
where, for each party indexed by $i$, we define
\begin{eqnarray}\label{Eq:Multi-eta}
\eta_i^\epsilon\coloneqq\gamma_i + \frac{\epsilon}{1-\epsilon}\left(\gamma_i - \frac{\id_i}{2}\right).
\end{eqnarray}
Write $\gamma_i = P^i_0\proj{0}+P^i_1\proj{1}$. Then, by the same reasoning as for Lemma~\ref{Lemma:LT}, we require 
\begin{eqnarray}\label{Eq:Multi-epsilon}
0\le\epsilon\le2\min_i{P^i_{1}}
\end{eqnarray}
in order to make sure $\eta_i^\epsilon$'s are all quantum states.

To see why $\mathcal{E}^{(\epsilon,N)}$ is a local thermalization, one can use Eq.~\eqref{Eq:Multi-eta} and note that ${\rm tr}_{\backslash i}[\mathcal{T}_{\rm GHZ}(\rho)] = \frac{\id_i}{2}$ for all $i$.
Furthermore, $\bra{\Psi_0^+}\mathcal{E}^{(\epsilon,N)}(\rho)\ket{\Psi_0^+}\ge\epsilon\bra{\Psi_0^+}\rho\ket{\Psi_0^+}.$
Since high enough overlap with $\ket{\Psi_0^+}$ implies genuinely multipartite entanglement, we conclude that for high enough local temperatures, by setting $\epsilon=2\min_i{P^i_{1}}$, we achieve a genuinely multipartite EPLT.

\section{Proof of Theorem~\ref{prop:equivalence}}\label{App:Eqv}
\begin{proof}
First, we note that every local bath thermalization is by definition a local thermalization.
To prove the inverse statement, we recall that LOSR channels are of the form $\int_\Lambda\left(\mathcal{E}_{\rm A}^\lambda\otimes\mathcal{E}_{\rm B}^\lambda\right)p_\lambda d\lambda$, which is in the convex hull of the set of all product channels .
Being embedded in a finite Euclidean space, Carath\'eodory theorem implies that for each LOSR channel $\mathcal{E}$, there exists a finite set of product channels and a probability distribution $\{\mathcal{E}_{\rm A}^i\otimes\mathcal{E}_{\rm B}^i,p_i>0\}_{i=1}^D$ such that $\mathcal{E} = \sum_{i=1}^Dp_i(\mathcal{E}_{\rm A}^i\otimes\mathcal{E}_{\rm B}^i)$, where $D$ only depends on the local dimensions.
Then, for a given $i$ and ${\rm X = A,B}$, the Stinespring dilation theorem~\cite{Wolf} guarantees the existence of an ancillary space ${\rm X}'_i$ with dimension $d^2$ and a unitary operator $U_{{\rm X}{\rm X}'_i}$ acting on ${\rm X}{\rm X}'_i$ such that $\mathcal{E}_{\rm X}^i(\cdot) = {\rm tr}_{{\rm X}'_i}\left\{U_{{\rm X}{\rm X}'_i}\left[(\cdot)\otimes{\proj{0}}_{{\rm X}'_i}\right]U_{{\rm X}{\rm X}'_i}^\dagger\right\}$.
Since ${\rm X}'_i\simeq\mathbb{C}^{d^2}$ for all $i$, we can simply choose them to be the same Hilbert space, denoted by ${\rm X'}\simeq\mathbb{C}^{d^2}$, and write the corresponding unitary operator as $U_{\rm XX'}^i$.
Then we have
\begin{eqnarray}
\mathcal{E}(\cdot) = {\rm tr}_{\rm A'B'}\left\{\sum_{i=1}^D p_i\left(U_{\rm AA'}^i\otimes U_{\rm BB'}^i\right)\left[(\cdot)\otimes{\proj{00}}_{\rm A'B'}\right]\left(U_{\rm AA'}^i\otimes U_{\rm BB'}^i\right)^\dagger\right\}.
\end{eqnarray}
Now we define a space $\mathcal{H}_D \coloneqq{\rm span}\left\{\ket{i}\right\}_{i=1}^D$, 
and we introduce two additional ancillary spaces ${\rm A''} \simeq\mathcal{H}_D$ and ${\rm B''}$ $\simeq\mathcal{H}_D$. 
Then we can write
\begin{eqnarray}
\mathcal{E}(\cdot) = {\rm tr}_{\rm A'B'A''B''}\left\{\left(V_{\rm AA'A''}\otimes V_{\rm BB'B''}\right)\left[(\cdot)\otimes{\proj{00}}_{\rm A'B'}\otimes\sum_{i=1}^Dp_i\proj{ii}_{\rm A''B''}\right]\left(V_{\rm AA'A''}\otimes V_{\rm BB'B''}\right)^\dagger\right\},
\end{eqnarray}
where
\begin{eqnarray}
V_{\rm XX'X''}\coloneqq\sum_i U_{\rm XX'}^i\otimes{\proj{i}}_{\rm X''},
\end{eqnarray}
which is a unitary operator acting on ${\rm XX'X''}$.
The separable state $\sum_{i=1}^Dp_i\proj{ii}_{\rm A''B''}$ is full rank, hence it can be identified with a thermal state on ${\rm A''B''}$ by an appropriate choice of the Hamiltonian on these ancillas.
\end{proof}
This result shows that, as it is intuitive, the set of local bath thermalizations coincides with the set of local thermalizations.

\section{Speed-Up For Infinite Twirling Time}\label{App:t_u}
In this section, we will go through the detailed proof of the speed-up result with infinite twirling time, which is Eq.~\eqref{Eq:Speed-Up} (and the statement above it).
The strategy is to consider a finite-time one-shot approximate implementation of the twirling operation.
Then we replace the exact twirling operation in Eq.~\eqref{Eq:AlternativeEPLT} by this approximation, and then compute the realization time of {\em $\delta$-thermalization}: Here we say a channel $\Lambda$ {\em $\delta$-thermalizes} a state $\rho$ to a thermal state $\gamma$ if \mbox{$\norm{\Lambda(\rho) - \gamma}_\infty<\delta$}.

As the first step, we recall that the twirling operation is defined as~\cite{Bennett1996,Horodecki1999}
\begin{eqnarray}
\mathcal{T}(\cdot)\coloneqq\int_{U(d)}(U\otimes U^*)(\cdot) (U\otimes U^*)^\dagger dU,
\end{eqnarray}
where the average is over the Haar measure. 
We will consider the implementation of $\mathcal{T}$ by means of a finite sequence of unitaries as introduced in Ref.~\cite{Toth2007},
\begin{eqnarray}\label{Eq:ImplementationT}
\mathcal{T}^{(N)}_{\bf U}\coloneqq\prod_{k=1}^N\mathcal{T}_k,
\end{eqnarray}
with
\begin{eqnarray}
\mathcal{T}_k(\cdot)\coloneqq\frac{1}{2}\mathcal{I}(\cdot) + \frac{1}{2}(U_k\otimes U_k^*)(\cdot)(U_k\otimes U_k^*)^\dagger,
\end{eqnarray}
where each $U_k$ represents a random unitary and ${\bf U} = (U_1,...,U_N)$ is a vector of random variables.
Setting \mbox{$\norm{\mathcal{E}}_\infty\coloneqq\sup_\rho\norm{\mathcal{E}(\rho)}_\infty$} for a given channel $\mathcal{E}$, it was proven in Ref.~\cite{Toth2007} that
\begin{eqnarray}\label{Eq:T-Estimate}
\left<\norm{\mathcal{T} - \mathcal{T}^{(N)}_{\bf U}}_\infty^2\right><\frac{1}{2^N},
\end{eqnarray}
where $\left<(\cdot)\right>\coloneqq\int(\cdot)dU_1dU_2...dU_N$ is the average over the Haar measure. This follows from Eq.~(22) of Ref.~\cite{Toth2007} and the fact that the sup norm is upper bounded by the other $p$ norms.
In order to establish the speed-up result, we will first give a more detailed version of the result in Ref.~\cite{Toth2007} by assessing the probability that a given realization of $\mathcal{T}^{(N)}_{\bf U}$ is close to $\mathcal{T}$:
\begin{alemma}\label{AppLemma}
For every $\lambda>0$, we have
\begin{eqnarray}
{\rm Prob}\left( \norm{\mathcal{T} - \mathcal{T}^{(N)}_{\bf U}}_\infty^2 - \frac{1}{2^N} > \lambda \right) < \frac{1}{\lambda^22^N}.
\end{eqnarray}
\end{alemma}
\begin{proof}
This fact can be seen by applying Chebyshev's inequality on the random variable $\norm{\mathcal{T} - \mathcal{T}^{(N)}_{\bf U}}_\infty^2$, whose variance can be shown to be upper bounded by $\frac{1}{2^N}$ via direct computation. 
To see this, we let $\Delta\coloneqq\norm{\mathcal{T} - \mathcal{T}^{(N)}_{\bf U}}_\infty^2$ and computation shows
\begin{eqnarray}
\left<\left(\Delta - \left<\Delta\right>\right)^2\right> = \left<\Delta^2\right> - \left<\Delta\right>^2\le\left<\Delta^2\right>\le\left<\norm{\mathcal{T} - \mathcal{T}^{(N)}_{\bf U}}_2^2\right><\frac{1}{2^N},
\end{eqnarray}
where $\norm{\mathcal{E}}_2\coloneqq\sup_\rho\norm{\mathcal{E}(\rho)}_2$ for the Hilbert-Schmidt norm $\norm{\cdot}_2$, and
the last inequality is due to the relation $\left<\norm{\left(\mathcal{T} - \mathcal{T}^{(N)}_{\bf U}\right)(\rho)}_2^2\right> = \frac{1}{2^N}\left(\norm{\rho}_2^2 - \norm{\mathcal{T}(\rho)}_2^2\right)<\frac{1}{2^N}$ given by Eq.~(22) in Ref.~\cite{Toth2007}.
Hence, the only thing to be checked is the applicability of Chebyshev's inequality, which requires the given random variable to be integrable.
It suffices to show the continuity of $\norm{\mathcal{T} - \mathcal{T}^{(N)}_{\bf U}}_\infty$ in the argument ${\bf U}= (U_1,...,U_N)$ with respect to the metric $d_N$ defined by $d_N({\bf U},{\bf V})\coloneqq\sum_{i=1}^N\norm{U_i - V_i}_\infty$.

Consider a given pair of sequences of unitaries ${\bf U}$ and ${\bf V}$.  Using the notation $\mathcal{U}_{i}(\cdot)\coloneqq(U_i\otimes U_i^*)(\cdot)(U_i\otimes U_i^*)^\dagger$, $\mathcal{V}_{i}(\cdot)\coloneqq(V_i\otimes V_i^*)(\cdot)(V_i\otimes V_i^*)^\dagger$, we get
\begin{eqnarray}
\left| \norm{\mathcal{T} - \mathcal{T}^{(N)}_{\bf U}}_\infty - \norm{\mathcal{T} - \mathcal{T}^{(N)}_{\bf V}}_\infty \right| \le \norm{\mathcal{T}^{(N)}_{\bf U} - \mathcal{T}^{(N)}_{\bf V}}_\infty\le\frac{1}{2^N}\sum_{\bf s}\norm{\prod_{i=1}^{j_{\bf s}}\mathcal{U}_{s_i} - \prod_{i=1}^{j_{\bf s}}\mathcal{V}_{s_i}}_\infty,
\end{eqnarray}
where we repeatedly used the triangle inequality and the summation $\sum_{\bf s}$ is over all the possible strings of ordered indices ${\bf s} = \{s_1,s_2,...,s_{j_{\bf s}}\}\subseteq\{1,2,...,N\}$ with $j_{\bf s}\le N$.
Since
\begin{eqnarray}
&\norm{\prod_{i=1}^{j_{\bf s}}\mathcal{U}_{s_i} - \prod_{i=1}^{j_{\bf s}}\mathcal{V}_{s_i}}_\infty& = \norm{\mathcal{U}_{s_{j_{\bf s}}}\circ\prod_{i=1}^{j_{\bf s}-1}\mathcal{U}_{s_i} - \mathcal{V}_{s_{j_{\bf s}}}\circ\prod_{i=1}^{j_{\bf s}-1}\mathcal{V}_{s_i}}_\infty\nonumber\\
&&\le\norm{\mathcal{U}_{s_{j_{\bf s}}}\circ\left(  \prod_{i=1}^{j_{\bf s}-1}\mathcal{U}_{s_i} - \prod_{i=1}^{j_{\bf s}-1}\mathcal{V}_{s_i}   \right)}_\infty + \norm{\left(    \mathcal{U}_{s_{j_{\bf s}}} - \mathcal{V}_{s_{j_{\bf s}}}    \right)\circ\prod_{i=1}^{j_{\bf s}-1}\mathcal{V}_{s_i}}_\infty\nonumber\\
&&  = \norm{\prod_{i=1}^{j_{\bf s}-1}\mathcal{U}_{s_i} - \prod_{i=1}^{j_{\bf s}-1}\mathcal{V}_{s_i}}_\infty + \norm{\mathcal{U}_{s_{j_{\bf s}}} - \mathcal{V}_{s_{j_{\bf s}}}}_\infty,
\end{eqnarray}
we conclude that
\begin{eqnarray}
\norm{\prod_{i=1}^{j_{\bf s}}\mathcal{U}_{s_i} - \prod_{i=1}^{j_{\bf s}}\mathcal{V}_{s_i}}_\infty\le\sum_{i=1}^{j_{\bf s}}\norm{\mathcal{U}_{s_i} - \mathcal{V}_{s_i}}_\infty.
\end{eqnarray}
 The continuity in the argument ${\bf U}$ in the metric $d_N$ follows from the fact that 
\begin{eqnarray}
&\norm{\mathcal{U}_i - \mathcal{V}_i}_\infty  &\leq  \norm{(U_i \otimes U^*_i - V_i \otimes V^*_i) \rho (U_i \otimes U^*_i)^\dag}_\infty + \norm{-(V_i \otimes V^*_i)  \rho (V_i \otimes V^*_i - U_i \otimes U^*_i)^\dagger}_\infty \nonumber\\
&&\leq 2 \norm{ U_i \otimes U^*_i- V_i \otimes V^*_i}_\infty \leq 2(\| U_i \otimes \id \|_\infty \| \id\otimes(U_i^* - V^*_i) \|_\infty + \| (U_i - V_i)\otimes\id \|_\infty \| \id \otimes V_i^* \|_\infty)\nonumber\\
&& = 4 \norm{ U_i - V_i}_\infty,
\end{eqnarray}
where in the first step we added and subtracted $(V_i \otimes V^*_i) \rho (U_i \otimes U^*_i)^\dag$ and used the triangle inequality; in the second step, we used the fact that for any two linear operators $A$ and $B$, $\norm {AB}_\infty \leq \norm {A}_\infty \norm {B}_\infty$ (submultiplicativity); in the third step, we added and subtracted $U_i \otimes V_i^*$ and again used the triangle inequality and submultiplicativity; and in the last step we used $\| A \otimes \id \|_\infty = \| A \|_\infty$.
\end{proof}
The above lemma implies that, for an arbitrarily small $\lambda$, with probability $1- O(\lambda^{-2} e^{-N})$, the realization $\mathcal{T}^{(N)}_{\bf U}$ is $\lambda + O(e^{-N})$ close to $\mathcal{T}$.

Now we formally state and prove the speed-up result, which is Eq.~\eqref{Eq:Speed-Up} (and the statement above it).
We recall the following notations: In the subsystem ${\rm X}$, we consider the following implementation of EPLT:
\begin{eqnarray}
\pe_{\rm X}^{({N},{t})} = \mathcal{P}_{\eta_{\rm X}^{dP_{\rm min}^{\rm X}}}^{t_{\rm X}}\circ\mathcal{T}^{(N)}_{{\bf U},{\rm X}},
\end{eqnarray} where $\mathcal{T}^{(N)}_{{\bf U},{\rm X}}\coloneqq\prod_{k=1}^N\left({\rm tr}_{\setminus \rm X}\circ\mathcal{T}_k\right)$ and $\mathcal{P}_\gamma^t(\cdot)\coloneqq e^{-\frac{t}{\tau_{\gamma}}}(\cdot) + \left(1-e^{-\frac{t}{\tau_{\gamma}}}\right)\gamma$ is given by Eq.~\eqref{Eq:PTM}.
Also, \mbox{$\eta_{\rm X}^{\epsilon_{\rm X}}\coloneqq \gamma_{\rm X} + \frac{\epsilon_{\rm X}}{1-\epsilon_{\rm X}} \left(\gamma_{\rm X} - \frac{\id_{\rm X}}{d}\right)$}.
\begin{theorem}
Let $\gamma_{\rm X}$ be the local thermal state.
If 
\begin{eqnarray}
\tau_{\gamma_{\rm X}} > t_U\times\frac{8}{\ln2},
\end{eqnarray}
then for every $\rho_{\rm X}\neq\gamma_{\rm X}$ and $p_*\in(0,1)$, there exists $\delta'>0$ such that for every $\delta\in(0,\delta')$:
\begin{enumerate}
\item there exists an integer $N_\delta\coloneqq\left\lceil8\log_2{\frac{d^2P_{\rm min}^{\rm X}\sqrt{2}}{\delta}}\right\rceil$ and a time $t_1$ such that $\Lambda_{\rm X}^{(N_\delta,t_1)}$ $\delta$-thermalizes $\rho_{\rm X}$ to $\gamma_{\rm X}$ with success probability 
\begin{eqnarray}
1- \left[ \frac{\delta}{d^2 P_{\rm min}^{\rm X} \sqrt{2}} \right]^4,
\end{eqnarray}
which can be chosen to be larger than $p_*$. 
\item $\mathcal{P}_{\gamma_{\rm X}}^t$ $\delta$-thermalizes $\rho_{\rm X}$ to $\gamma_{\rm X}$ only if $t\ge t_2$.
\item $t_1<t_2$.
\end{enumerate}
\end{theorem}
\begin{proof}
We use the shortcut notation $\eta_{\rm X}$ for the state $\eta_{\rm X}^{d P_{\rm min}^{\rm X}}$.
To get an explicit estimate on time, we take $\lambda = 2^{-\frac{N}{4}}$ and use Lemma~\ref{AppLemma}.
By noting that $t_{\rm X} = -\tau_{\eta_{\rm X}}\ln{d P_{\rm min}^{\rm X}}$, we start with the computation of the local state,
\begin{eqnarray}
&\norm{\pe_{\rm X}^{({N},{t})}(\rho_{\rm X}) - \gamma_{\rm X}}_\infty&=\norm{\mathcal{P}_{\eta_{\rm X}}^{t_{\rm X}}\circ\mathcal{T}^{(N)}_{{\bf U},{\rm X}}(\rho_{\rm X}) - \gamma_{\rm X}}_\infty
\nonumber\\
&&=\norm{d P_{\rm min}^{\rm X}\times\mathcal{T}^{(N)}_{{\bf U},{\rm X}}(\rho_{\rm X}) + \left(1-d P_{\rm min}^{\rm X}\right)\times\eta_{\rm X} - \gamma_{\rm X}}_\infty\nonumber\\
&&= d P_{\rm min}^{\rm X}\norm{\mathcal{T}^{(N)}_{{\bf U},{\rm X}}(\rho_{\rm X}) - \frac{\id_{\rm X}}{d}}_\infty\nonumber\\
&&\le d^2P_{\rm min}^{\rm X}\norm{\mathcal{T}^{(N)}_{\bf U} - \mathcal{T}}_\infty\nonumber\\
&&<d^2P_{\rm min}^{\rm X}\sqrt{2}\times 2^{-\frac{N}{8}},
\end{eqnarray}
which holds with probability $1 - 2^{-\frac{N}{2}}$.
In the first inequality, we used the relation \mbox{$\norm{{\rm tr}_{\rm Y}(\cdot)}_\infty\le\norm{{\rm tr}_{\rm Y}(\cdot)}_1\le\norm{\cdot}_1\le d\norm{\cdot}_\infty$} and $\norm{\mathcal{Q}(\rho)}_\infty\le\norm{\mathcal{Q}}_\infty$ for all superoperators $\mathcal{Q}$ and states $\rho$; in the second inequality we used \mbox{$\sqrt{\lambda + 2^{-N}}<\sqrt{2\lambda} = \sqrt{2}\times2^{-\frac{N}{8}}$.}
This estimate means that for any given $\delta\in(0,1)$, there exists a sufficiently large $N=N_\delta$ to let the above upper bound be smaller than $\delta$; that is, this choice of $N$ ensures $\delta$-thermalization of {\em every} local input $\rho_{\rm X}$, with success probability $1 - 2^{-\frac{N}{2}}$. 
It suffices to take
\begin{eqnarray}\label{Eq:N_delta}
N_\delta\coloneqq\left\lceil8\log_2{\frac{d^2P_{\rm min}^{\rm X}\sqrt{2}}{\delta}}\right\rceil,
\end{eqnarray}
where $\lceil x\rceil$ is the smallest integer larger than $x$.

Now consider a given $\delta\in(0,1)$ and a given local input state $\rho_{\rm X}$.
Then, $t_{\rm X} = -\tau_{\eta_{\rm X}}\ln{d P_{\rm min}^{\rm X}}$ and $N_\delta$ gives us the following total implementation time of the channel $\pe_{\rm X}^{({N_\delta},{t_1})}$:
\begin{eqnarray}\label{Eq:N-EPLT-time}
t_1 = t_{\rm X} + N_\delta t_U = \tau_{\eta_{\rm X}}\ln{\frac{1}{dP_{\rm min}^{\rm X}}} + N_\delta t_U.
\end{eqnarray}

Now, if Alice and Bob simply leave their local systems in contact with local independent baths, the partial thermalization model $\delta$-thermalizes the local state $\rho_{\rm X}$ in a time
\begin{eqnarray}\label{Eq:PT-Time}
t_2 = \tau_{\gamma_{\rm X}}\ln{\frac{\norm{\rho_{\rm X}-\gamma_{\rm X}}_\infty}{\delta}}.
\end{eqnarray}

Combining Eqs.~\eqref{Eq:N_delta},~\eqref{Eq:N-EPLT-time}, and~\eqref{Eq:PT-Time}, we learn that
$t_1<t_2$, with probability $1 - 2^{-\frac{N_\delta}{2}}$, if $0<\tau_{\gamma_{\rm X}}\ln{\frac{\norm{\rho_{\rm X}-\gamma_{\rm X}}_\infty}{\delta}} +  \tau_{\eta_{\rm X}}\ln{(dP_{\rm min}^{\rm X})} - N_\delta t_U.$
This is true if
\begin{eqnarray}\label{Eq:SuffCondi}
\norm{\rho_{\rm X}-\gamma_{\rm X}}_\infty > f\times  \delta^{\left(1 - \frac{t_U}{\tau_{\gamma_{\rm X}}}\frac{8}{\ln{2}}\right)},
\end{eqnarray}
where $f\coloneqq(dP_{\rm min}^{\rm X})^{-\frac{\tau_{\eta_{\rm X}}}{\tau_{\gamma_{\rm X}}}}e^{\frac{t_U}{\tau_{\gamma_{\rm X}}}}\left(d^2P_{\rm min}^{\rm X}\sqrt{2}\right)^{\frac{t_U}{\tau_{\gamma_{\rm X}}}\frac{8}{\ln{2}}}$ is a constant in $\delta$.
Note that $f$ is finite for all possible values of $\frac{\tau_{\eta_{\rm X}}}{\tau_{\gamma_{\rm X}}}$.
This means that when the exponent of $\delta$ in Eq.~\eqref{Eq:SuffCondi} is positive, it is always possible to find a small enough $\delta$ to achieve Eq.~\eqref{Eq:SuffCondi}.
Specifically, suppose
\begin{eqnarray}
\tau_{\gamma_{\rm X}} > t_U\times\frac{8}{\ln2}\approx t_U\times 11.5416.
\end{eqnarray}
Then, for any given $\rho_{\rm X}\neq\gamma_{\rm X}$, a successful implementation of twirling will demonstrate $t_1<t_2$ (i.e., a speed-up effect) for every $\delta>0$ small enough, where the success probability is given by
\begin{eqnarray}
1- \left[ \frac{\delta}{d^2 P_{\rm min}^{\rm X} \sqrt{2}} \right]^4.
\end{eqnarray}
This completes the proof.
\end{proof}
Since in most cases $t_U \ll \tau_{\gamma_X}$, $\delta$-thermalization with small enough $\delta $ is faster in the EPLT than in the standard thermalization model, even taking into account the time to implement the random unitaries.


\begin{thebibliography}{99}
\bibitem{footnote1} One might wonder why we do not simply do a sudden quench of the local Hamiltonians so that the local state is thermal. However, note that this protocol depends on the input state. In other words, it does not satisfy condition~\ref{condition2} in the definition of local thermalization.
\bibitem{footnote4} More precisely, the completion time is finite except for the trivial cases $\tau_{\eta_{\rm X}^{\epsilon_{\rm X}}}=\infty$ and $\epsilon_{\rm X}=0$. The first is unphysical and the second does not give an EPLT.
\end{thebibliography}
\end{document}